\documentclass[journal,onecolumn,12pt, draftclsnofoot]{IEEEtran}
\usepackage{amsfonts,color,morefloats}
\usepackage{booktabs}
\usepackage{amssymb,amsmath,latexsym,amsthm}
\usepackage{graphicx}
\usepackage{makecell}
\usepackage{hyperref}
\usepackage{multirow}
\usepackage{array}

\newtheorem{theorem}{Theorem}[section]
\newtheorem{lemma}[theorem]{Lemma}

\newtheorem{example}{Example}

\newcommand{\C}{\mathcal{C}}

\begin{document}
	
	\title{Parameters of several families of binary duadic codes and their related codes
\thanks{The work of Hai Liu and Chengju Li was supported by the National
Natural Science Foundation of China (12071138), Shanghai Natural Science Foundation (22ZR1419600), the open research fund of National Mobile Communications Research Laboratory of Southeast University (2022D05). The work of Haifeng Qian was supported by the Innovation Program
of Shanghai Municipal Education Commission (2021-01-07-00-08-E00101), and ``Digital Silk Road'' Shanghai International Joint Lab of Trustworthy Intelligent Software (22510750100).
}}
\author{Hai Liu,~~ Chengju Li, ~~ Haifeng Qian
\thanks{H. Liu and C. Li are with MoE Engineering Research Center of Software/Hardware Co-design Technology and Application, East China Normal University, Shanghai, 200062, China; and are also with the National Mobile Communications Research Laboratory, Southeast University, Nanjing 210096, China
(email: 52265902014@stu.ecnu.edu.cn, cjli@sei.ecnu.edu.cn).}
\thanks{H. Qian is with the School of Software Engineering East
China Normal University, Shanghai 200062, China
(email: hfqian@admin.ecnu.edu.cn).}}
	\date{\today}
	\maketitle
	
	\begin{abstract}
	 Binary duadic codes are an interesting subclass of cyclic codes since they have large dimensions and their minimum distances may have a square-root bound.
 In this paper, we present several families of binary duadic codes of length $2^m-1$ and develop some lower bounds on their minimum distances by using the BCH bound on cyclic codes, which partially solves one case of the open problem proposed in \cite{LLD}. It is shown that the lower bounds on their minimum distances are close to the square root bound.
Moreover, the parameters of the dual and extended codes of these
	 binary duadic codes are investigated.
	\end{abstract}

		\section{Introduction} \label{sec-intro}
	
	In this paper, let $\Bbb F_q $ denote the finite field of order $q$, where $q$ is a power of a prime $p$.
	An $[n,k,d]$ linear code $\mathcal C$ over $\Bbb F_q$ is a $k$-dimensional subspace of $\Bbb F_q^n$ with minimum (Hamming) distance $d$. The dual code of $\mathcal{C}$, denoted by $\mathcal{C}^{\perp}$, is defined by
	$$\mathcal{C}^{\perp}=\{\mathbf{b} \in\Bbb F_q^n \ :\ \mathbf{b}\mathbf{c}^{T}=0 \,\,\text{for all $\mathbf{c} \in \mathcal{C}$}\}, $$
	where $\mathbf{b}\mathbf{c}^{T}$ is the standard inner product of two vectors $\mathbf{b}$ and $\mathbf{c}$ in $\Bbb F_q^n$. In addition, define the extended code $\bar{\mathcal C}$ to be the code
$$\bar{\mathcal C}=\{(c_0 , c_1, \ldots, c_{n-1}, c_n):(c_0 , c_1, \ldots, c_{n-1}) \in \mathcal{C} \text{ with } c_0 +c_1+\cdots+ c_{n-1}+ c_n=0\}.$$ It is easy to see that $\bar{\mathcal C}$ is an $[n+1, k]$ linear code.

	The linear code $\mathcal C$ over $\Bbb F_q$ is said to be \emph{cyclic} if $\mathbf (c_0, c_1, \ldots, c_{n-1}) \in \mathcal C $ implies $\mathbf (c_{n-1}, c_0, \ldots, c_{n-2}) \in \mathcal C $.
	By identifying each vector $\mathbf (c_0, c_1, \ldots, c_{n-1}) \in \Bbb F_q^n $ with
	$$c_0+c_1x+c_2x^2+\cdots+c_{n-1}x^{n-1} \in \Bbb F_q[x]/\langle x^n-1 \rangle,$$
	a code $\mathcal C$ of length $n$ over $\Bbb F_q$ corresponds to a subset of $\Bbb F_q[x]/\langle x^n-1 \rangle$. Then $\mathcal C$ is a cyclic code if and only if the corresponding subset is an ideal of $\Bbb F_q[x]/\langle x^n-1 \rangle$.
	Note that every ideal of $\Bbb   F_q[x]/\langle x^n-1 \rangle$ is principal. Then there is a monic polynomial $g(x)$ of the smallest degree such that $\mathcal C=\langle g(x) \rangle$ and $g(x) \mid (x^n-1)$. Then $g(x)$ is called the \emph{generator polynomial} and $h(x)=(x^n-1)/g(x)$ is referred to as the \emph{check polynomial} of $\mathcal C$.
	Throughout this paper, assume that $\gcd(q,n)=1$.
	Denote $m=\text{ord}_n(q)$, i.e., $m$ is the smallest positive integer such that $q^m \equiv 1 \pmod n$.
	Let $\alpha$ be a primitive element of $\Bbb F_{q^m}$ and put $\beta=\alpha^{\frac {q^m-1} n}$. Then
	$\beta$ is a primitive $n$-th root of unity. The set $T=\{0 \le i \le n-1 : g(\beta^i)=0\}$ is referred to as the
	\emph{defining set} of $\mathcal C$ with respect to $\beta$. If $T$ contains $\delta-1$ consecutive integers, then we have the
	well-known BCH bound on the minimum distance of cyclic codes, i.e., $d \geq \delta$.

 	 Let $S_1$ and $S_2$
    be two subsets of $\Bbb Z_n=\{0,1,2,\ldots,n-1\}$ such that
    \begin{itemize}
    	\item $S_1 \cap S_2 = \emptyset$ and $S_1 \cup S_2=\Bbb Z_n \setminus \{0\}$, and
    	\item both $S_1$ and $S_2$ are the union of some $2$-cyclotomic cosets modulo $n$.
    \end{itemize}
    If there is a unit $\mu \in \Bbb Z_n$ such that $S_1 \mu = S_2$ and $S_2 \mu =S_1$,
    then $(S_1, S_2, \mu)$ is called a \emph{splitting\index{splitting}} of $\Bbb Z_n$.

    Let $(S_1, S_2, \mu)$ be a
    splitting of $\Bbb Z_n$.
    Define
    $$
    g_i(x)=\prod_{i \in S_i} (x - \beta^i) \ \mbox{ and } \ \tilde{g}_i(x)=(x-1) g_i(x)
    $$
    for $i \in \{1,2\}$. The pair of cyclic codes $\mathcal C_1$ and $\mathcal C_2$ of length $n$ over $\Bbb F_2$ with generator
    polynomials $g_1(x)$ and $g_2(x)$ are called \emph{odd-like duadic codes\index{odd-like duadic codes}},
    and the pair of cyclic codes $\tilde{\mathcal C}_1$ and $\tilde{\mathcal C}_2$ of length $n$ over $\Bbb F_2$ with generator
    polynomials $\tilde{g}_1(x)$ and $\tilde{g}_2(x)$ are called \emph{even-like duadic codes\index{even-like duadic codes}}.

    By definition, the binary cyclic codes $\mathcal C_1$ and $\mathcal C_2$ have parameters $[n, (n+1)/2]$ and the binary cyclic codes $\tilde{\mathcal C}_1$ and $\tilde{\mathcal C}_2$ have
	parameters $[n, (n-1)/2]$. Binary cyclic codes with parameters $[n, (n \pm 1)/ 2]$ were investigated in \cite{Ding12, HP03, TD, Xiong, XiongZhang}. It is observed that these binary cyclic codes have large dimensions. Generally,
it is very hard to determine the minimum distance of a cyclic code with parameters $[n, (n \pm 1)/2]$. For now, the best one can do is to develop a good lower bound on the minimum distance of the code.
For odd-like duadic codes, we have the following result \cite[Theorem 6.5.2]{HP03}.
	
	\begin{theorem}[Square root bound]\label{thm-srb}
		Let $\mathcal C_1$ and $\mathcal C_2$ be a pair of odd-like duadic codes of length $n$ over $\Bbb F_2$. Let $d_o$ be their
		(common) minimum odd weight. Then the following hold:
		\begin{enumerate}
			\item $d_o^2 \ge n$.
			\item If the splitting defining the duadic codes is given by $\mu=-1$, then $d_o^2-d_o+1 \geq n$.
			\item Suppose $d_o^2-d_o+1 = n$, where $d_o >2$, and assume that the splitting defining the duadic codes is given by $\mu=-1$. Then $d_o$ is the minimum weight of both $\mathcal C_1$ and $\mathcal C_2$.
		\end{enumerate}
	\end{theorem}
	
As a generalisation of the quadratic residue codes, duadic codes were introduced
and investigated in \cite{Leon84, Leon88, Ples87},
 where a number of properties are proved. In addition, Pless, Masley and Leon presented all binary duadic codes of length
until 241 \cite{Ples87}. The total number of binary duadic codes of prime power lengths and their constructions were
presented in  \cite{DLX99} and \cite{DingPless99}.
For more information on duadic codes, the reader is referred to \cite[Chapter 6]{HP03}.

Let $m \geq 2$ be a positive integer and let $n = {2^m -1}$.
 For any $s \in \Bbb Z_n$, the $2$-cyclotomic coset of $s$ modulo $n$ is defined by
	$$C_s^{(2,n)}=\{s,s2,s2^2,\ldots,s2^{l_s-1}\}\bmod n \subseteq \Bbb Z_n,$$
	where $l_s$ is the smallest positive integer such that $s \equiv s2^{l_s}\pmod  n$.
For an integer $i$ with $0 \le i \le 2^m-1$, let $$i=i_{m-1}2^{m-1} + i_{m-2}2^{m-2} + \cdots + i_1 2+i_0$$ be the $2$-adic expansion of $i$, where
	$i_j \in \{0, 1\}$ with $0 \le j \le m-1$.
For any $i$ with $ 0 \leq i \leq n-1 $, define $w_{2}(i) = \sum\limits_{j=0}^{m-1}i_{j}$. 	

Now we recall a construction of binary duadic codes documented in \cite{LLD}.
 Let $r \geq 2$ be a positive integer and let $n=2^m-1$ for an integer $m \geq 3$.
	 Let $S$ be any proper subset of $\mathbb{Z}_r$.
     Define
     $$
     T_{[r,m,S]}=\{1 \leq i \leq n-1: w_2(i) \bmod{r} \in S\}.
     $$	
     By definition, $T_{[r,m,S]}$ is the union of some $2$-cyclotomic cosets modulo $n$. Let $\alpha$ be a primitive element of
     $\mathbb{F}_{2^m}$.
	 Let $\C_{[r,m,S]}$ denote the binary cyclic code of length $n$ with generator polynomial
	 $$
	 g_{[r,m,S]}(x)=\prod_{i \in T_{[r,m,S]}} (x-\alpha^i).
	 $$
Let $r \geq 2$ be even and $|S|=r/2$, the code $\C_{[r,m,S]}$ could be a duadic code for certain odd $m$.
\begin{enumerate}
  \item When $r=2$ and $|S|=1$, a family of duadic codes $\C_{[r,m,S]}$ were studied in \cite{TD}.
  \item When $r=4$ and $|S|=2$, two families of duadic codes $\C_{[r,m,S]}$ were presented in \cite{LLD}, where the following open problem was also proposed.
\end{enumerate}
 \textbf{Open problem.} \cite{LLD}
	 	Let $r \geq 6$ be an even integer. Find a subset $S$ of $\mathbb{Z}_r$ with $|S|=r/2$
	 	such that $\C_{[r,m,S]}$ is a binary duadic code of length $n=2^m-1$ for infinitely many odd $m$.
	 	Determine the parameters of these duadic codes.

	 In this paper, we present binary duadic codes for all possible $S$ when $r=6$ and $|S|=3$, and develop some lower bounds on their minimum distances by using the BCH bound on cyclic codes. As will be seen,
 these lower bounds are close to the square root bound.
This partially solves one case of the open problem.
Moreover, the parameters of the dual and extended codes of these
	 binary duadic codes are also investigated.

	\section{Constructions of all binary duadic codes for $r=6$ }
	
In this section, we follow the notation specified in Section \ref{sec-intro}.
	 Let $S$ be any proper subset of $\mathbb{Z}_6$ with $|S|=3$ and $\bar S= \mathbb{Z}_6 \setminus S$.
	It is  clear that $T_{[6,m,S]}$ is the defining set of $\mathcal C_{[6,m,S]}$ with respect to the $n$-th primitive root of unity $\alpha$, where $n=2^m-1$.
	Note that $\Bbb Z_6^*=\{1, -1\}$. It then follows that $\mu=-1$ for a splitting of $\Bbb Z_6$.
Moreover, we have $\left|\{ S \subsetneq \Bbb Z_6: |S|=3\}\right|=\binom 6 3=20$. Then all
binary duadic codes for $r=6$ can be constructed as follows.

	\begin{enumerate}
		\item  When $m \equiv 1 \pmod 6$, let $S \in \Big\{\{0,2,3\}, \{0,2,4\}, \{0,4,5\}, \{0,3,5\}\Big\}$. Then $\bar S=-S$ and $S \cup \bar S=\Bbb Z_6$.
Note that $\omega_2(i)=m-\omega_2(n-i)$ for each $i$ with $1 \leq i \leq n-1$.
This leads to
$$T_{[6,m,S]}=-T_{[6,m,\bar S]}  \text{ and } T_{[6,m,S]} \cup T_{[6,m,\bar S]}=\Bbb Z_n \setminus \{0\}.$$
Thus
 $\mathcal \C_{[6,m,S]}$ and $\mathcal \C_{[6,m,\bar S]}$ form a  pair of odd-like duadic codes.
		
		\item  When $m \equiv 3 \pmod 6$, let $S \in \Big\{\{0,1,4\}, \{0,1,5\}, \{0,2,4\}, \{0,2,5\}\Big\}$. Then one can verify that
$\mathcal \C_{[6,m,S]}$ and $\mathcal \C_{[6,m,\bar S]}$ form a  pair of odd-like duadic codes.
		
		\item  When $m \equiv 5 \pmod 6$, let $S \in \Big\{\{0,1,2\}, \{0,1,3\}, \{0,2,4\}, \{0,3,4\}\Big\}$. Then one can verify that
 $\mathcal \C_{[6,m,S]}$ and $\mathcal \C_{[6,m,\bar S]}$ form a  pair of odd-like duadic codes.		
	\end{enumerate}

It is remarked that the parameters of the codes $\mathcal \C_{[6,m,S]}$ and $\mathcal \C_{[6,m,\bar S]}$ for $S=\{0,2,4\}$ had been studied in \cite{TD}.
Below we mainly investigate the parameters of the binary duadic codes $\mathcal \C_{[6,m,S]}$ and $\mathcal \C_{[6,m,\bar S]}$ except $S=\{0,2,4\}$ and their dual and extended codes.
Lower bounds on their minimum distances are presented, and they are close to the square root bounds.
	
	\section{Some auxiliary results}

In this section, we will present some necessary auxiliary results on the defining sets of binary duadic cyclic codes, which play an important role
in developing lower bounds on minimum distances of the binary codes.
The following well-known lemma will be employed later.
	

	\begin{lemma}\label{lemma-gcd2}
		Let $l$ and $m$ be two positive integers. Then
		$$\gcd(a^{m}-1, a^{l}-1) = a^{\gcd(m,l)} -1,$$
		where $a \geq 2$ is a positive integer.
	\end{lemma}

Below we divide into three cases according to the construction of the binary duadic codes.

	\vspace{5mm}
\emph {A. The case: $m \equiv 1 \pmod{6}$}	

When $S=\{0, 4, 5\}$, we have two subcases: $m \equiv 1 \pmod{12}$ and $m \equiv 7 \pmod{12}$.

	\begin{lemma}\label{lemma-m121}
		Let $m \equiv 1 \pmod{12} \geq 13$. Then we have the following.
\begin{enumerate}
  \item If $v=2^{(m-1)/2} -1$, then $\gcd(v,n)=1$ and
		$$\{av : 1 \leq a \leq 2^{(m-1)/2} +2 \} \subseteq	 T_{[6,m,\{0,4,5\}]}.$$
  \item If $v=2^{(m+1)/2} -1$, then $\gcd(v,n)=1$ and $$\{av : 1 \leq a \leq 2^{(m-1)/2} +2 \} \subseteq  T_{[6,m, \{1,2,3\}]}.$$
\end{enumerate}		
	\end{lemma}
	
	\begin{proof}
If $v=2^{(m-1)/2} -1$, it follows from Lemma \ref{lemma-gcd2} that $ \gcd(v,n) =1$.
		 When $a=2^{(m-1)/2} +2$, we have
		$$av =2^{m-1} + 2^{(m-1)/2} -2 = 2(2^{m-2} + 2^{(m-3)/2}-1).$$
		Consequently, $w_{2}(av) =(m-1)/2 \equiv 0 \pmod{6}$.
		When $ a=2^{(m-1)/2} +1 $, $ av = 2^{m-1} -1$ and $w_{2}(av) =(m-1)\equiv 0 \pmod{6}$.
		When $ a=2^{(m-1)/2}$, $w_{2}(av) =w_{2}(v)=(m-1)/2\equiv 0 \pmod{6}$.
		Now we assume that $1 \leq a \leq 2^{(m-1)/2} - 1$. Let $a=2^{l}\bar{a}$, where $\bar{a}$ is odd and $l \ge 0$ is an integer. Then we have $1 \leq a \leq 2^{(m-1)/2} - 1$ and the $2$-adic expansion of $\bar{a}$ given by $$ \bar{a} = \sum\limits_{i=0}^{(m-3)/2}a_{i}2^{i}.$$
		Since $\bar{a}$ is odd, $a_{0}=1$. We have
		\begin{align*}
			\bar{a}v =\sum\limits_{i=1}^{(m-3)/2}a_{i}2^{i+(m-1)/2} + \sum\limits_{i=0}^{(m-3)/2}(1-a_{i})2^{i} +1.
		\end{align*}
	    It then follows that
		$$w_{2}(\bar{a}v) =  w_{2}(\bar{a})-1 + 1+ \frac{m-1}{2}-w_{2}(\bar{a}) = \frac{m-1}{2}\equiv 0 \pmod{6}.$$
		The desired conclusion on the first case then follows.

		If $v=2^{(m+1)/2} -1$, it follows from Lemma \ref{lemma-gcd2} that
		$$\gcd(v,n)=2^{\gcd((m+1)/2,m)}-1 =2^{\gcd((m+1)/2,(m-1)/2)}-1 =1.$$
	When $ a=2^{(m-1)/2}$, it is easy to see that
		$$w_{2}(av) = w_{2}(v) = \frac{m+1}{2}\equiv 1 \pmod{6}.$$
Furthermore, one can similarly check that $w_{2}(av) \equiv 1 \pmod{6}$ for $a=2^{(m-1)/2}+1$ and $2^{(m-1)/2}+2$.
		Next, we assume that $1 \leq a \leq 2^{(m-1)/2} - 1$. Let $a=2^{l}\bar{a}$, where $\bar{a}$ is odd and $l \ge 0$ is an integer. Then we have $1 \leq \bar{a} \leq 2^{(m-1)/2} - 1$. Let the 2-adic expansion of $\bar{a}$ be given by
		$$\bar{a} = \sum\limits_{i=0}^{(m-3)/2}a_{i}2^{i}.$$
		Since $\bar{a}$ is odd, $a_{0}=1$.  Then
		\begin{align*}
			\bar{a}v &= \bar{a}2^{(m+1)/2} -\bar{a} \\
			&= \sum\limits_{i=1}^{(m-3)/2}a_{i}2^{i+(m+1)/2} +2^{(m-1)/2} + \sum\limits_{i=0}^{(m-3)/2}(1-a_{i})2^{i} +1.
		\end{align*}
		As a result, we have
		$$ w_{2}(\bar{a}v) =  w_{2}(\bar{a})-1 + 2 + \frac{m-1}{2}-w_{2}(\bar{a}) = \frac{m+1}{2}\equiv 1 \pmod{6}.$$
		This completes the proof.
	\end{proof}

\begin{lemma}\label{lemma-m127}
	Let $m \equiv 7 \pmod{12} \geq 7$. Then we have the following.
	\begin{enumerate}
		\item If $v=2^{(m+1)/2} -1$, then $\gcd(v,n)=1$ and
		$$\{av : 1 \leq a \leq 2^{(m-1)/2} \} \subseteq  T_{[6,m,\{0,4,5\}]}.$$
		\item If $v=2^{(m-1)/2} -1$, then $\gcd(v,n)=1$ and $$\{av : 1 \leq a \leq 2^{(m-1)/2} \} \subseteq  T_{[6,m,\{1,2,3\}]}.$$
	\end{enumerate}		
\end{lemma}
\begin{proof}
	The proof is very similar to that of Lemma \ref{lemma-m121} and omitted here.
\end{proof}

When $S=\{0,2,3\}$ and $\{0,3,5\}$, we have the following lemma on the defining sets $T_{[6,m,S]}$ and $T_{[6,m,\bar S]}$.

\begin{lemma}\label{lemma-m161}
	Let $m \equiv 1 \pmod{6} \geq 7$. Then we have the following.
	\begin{enumerate}
		\item If $v=2^{(m-1)/2} -1$, then $\gcd(v,n)=1$ and
		\begin{itemize}
			\item $\{av : 1 \leq a \leq 2^{(m-1)/2}+2  \} \subseteq   T_{[6,m,\{0,2,3\}]}$,
			\item $ \{av : 1 \leq a \leq 2^{(m-1)/2}+2  \} \subseteq  T_{[6,m,\{0,3,5\}]}$.
		\end{itemize}
		
		\item If $v=2^{(m+1)/2} -1$, then $\gcd(v,n)=1$ and
		\begin{itemize}
			\item $\{av : 1 \leq a \leq 2^{(m-1)/2}+2 \} \subseteq  T_{[6,m,\{1,4,5\}]}$,
			\item $\{av : 1 \leq a \leq 2^{(m-1)/2}+2 \} \subseteq  T_{[6,m,\{1,2,4\}]}$.
		\end{itemize}
		
	\end{enumerate}		
\end{lemma}
\begin{proof}
	The proof is very similar to that of Lemma \ref{lemma-m121} and omitted here.
\end{proof}

\vspace{5mm}
\emph {B. The case: $m \equiv 3 \pmod{6}$}	

		\begin{lemma}\label{lemma-m123}
		Let $m \equiv 3 \pmod{12} \geq 3$. Then we have the following.
 \begin{enumerate}
  \item If $v=2^{(m-1)/2} -1$, then $\gcd(v,n)=1$ and
		$$\{av : 1 \leq a \leq 2^{(m-1)/2}  \} \subseteq  T_{[6,m,\{0,1,5\}]}.$$
  \item If $v=2^{(m+1)/2} -1$, then $\gcd(v,n)=1$ and $$\{av : 1 \leq a \leq 2^{(m-1)/2} \} \subseteq  T_{[6,m,\{2,3,4\}]}.$$
\end{enumerate}		
	\end{lemma}
    \begin{proof}
    	The proof is very similar to that of Lemma \ref{lemma-m121} and omitted here.
    \end{proof}

  \begin{lemma}\label{lemma-m129}
  	Let $m \equiv 9 \pmod{12} \geq 9$. Then we have the following.
  	\begin{enumerate}
  		\item If $v=2^{(m+1)/2} -1$, then $\gcd(v,n)=1$ and
  		$$\{av : 1 \leq a \leq 2^{(m-1)/2} +2 \} \subseteq  T_{[6,m,\{0,1,5\}]}.$$
  		\item If $v=2^{(m-1)/2} -1$, then $\gcd(v,n)=1$ and $$\{av : 1 \leq a \leq 2^{(m-1)/2} +2 \} \subseteq  T_{[6,m,\{2,3,4\}]}.$$
  	\end{enumerate}		
  \end{lemma}
  \begin{proof}
  	The proof is very similar to that of Lemma \ref{lemma-m121} and omitted here.
  \end{proof}

  \begin{lemma}\label{lemma-m361}
  	Let $m \equiv 3 \pmod{6} \geq 9$. Then we have the following.
  	\begin{enumerate}
  		
  		\item If $v=2^{(m+1)/2} -1$, then $\gcd(v,n)=1$ and
  		\begin{itemize}
  			\item $\{av : 1 \leq a \leq 2^{(m-1)/2}  \} \subseteq  T_{[6,m,\{0,2,5\}]}$,
  			\item $\{av : 1 \leq a \leq 2^{(m-1)/2} \} \subseteq  T_{[6,m,\{0,1,4\}]}$.
  		\end{itemize}
  		\item If $v=2^{(m-1)/2} -1$, then $\gcd(v,n)=1$ and
  		\begin{itemize}
  			\item $\{av : 1 \leq a \leq 2^{(m-1)/2} \} \subseteq  T_{[6,m,\{1,3,4\}]}$,
  			\item $\{av : 1 \leq a \leq 2^{(m-1)/2} \} \subseteq  T_{[6,m,\{2,3,5\}]}$.
  		\end{itemize}  		
  	\end{enumerate}		
  \end{lemma}
  \begin{proof}
  	The proof is very similar to that of Lemma \ref{lemma-m121} and omitted here.
  \end{proof}

 \vspace{5mm}
\emph {C. The case: $m \equiv 5 \pmod{6}$}	
 	
		\begin{lemma}\label{lemma-m125}
		Let $m \equiv 5 \pmod{12} \geq 5$. Then we have the following.
 \begin{enumerate}
  \item If $v=2^{(m-1)/2} -1$, then $\gcd(v,n)=1$ and
		$$\{av : 1 \leq a \leq 2^{(m-1)/2}  \} \subseteq  T_{[6,m,\{0,1,2\}]}.$$
  \item If $v=2^{(m+1)/2} -1$, then $\gcd(v,n)=1$ and $$\{av : 1 \leq a \leq 2^{(m-1)/2} \} \subseteq  T_{[6,m,\{3,4,5\}]}.$$
\end{enumerate}		
	\end{lemma}
    \begin{proof}
    	The proof is very similar to that of Lemma \ref{lemma-m121} and omitted here.
    \end{proof}

   \begin{lemma}\label{lemma-m1211}
   	Let $m \equiv 11 \pmod{12} \geq 11$. Then we have the following.
   	\begin{enumerate}
   		\item If $v=2^{(m+1)/2} -1$, then $\gcd(v,n)=1$ and
   		$$\{av : 1 \leq a \leq 2^{(m-1)/2} +2 \} \subseteq  T_{[6,m,\{0,1,2\}]}.$$
   		\item If $v=2^{(m-1)/2} -1$, then $\gcd(v,n)=1$ and $$\{av : 1 \leq a \leq 2^{(m-1)/2} +2 \} \subseteq  T_{[6,m,\{3,4,5\}]}.$$
   	\end{enumerate}		
   \end{lemma}
   \begin{proof}
   	The proof is very similar to that of Lemma \ref{lemma-m121} and omitted here.
   \end{proof}

   \begin{lemma}\label{lemma-m561}
   	Let $m \equiv 5 \pmod{6} \geq 5$. Then we have the following.
   	\begin{enumerate}
   		\item If $v=2^{(m+1)/2} -1$, then $\gcd(v,n)=1$ and
   		\begin{itemize}
   			\item $\{av : 1 \leq a \leq 2^{(m-1)/2} +2  \} \subseteq  T_{[6,m,\{0,1,3\}]}$,
   			\item $\{av : 1 \leq a \leq 2^{(m-1)/2} \} \subseteq  T_{[6,m,\{0,3,4\}]}$.
   		\end{itemize}
   	
   		\item If $v=2^{(m-1)/2} -1$, then $\gcd(v,n)=1$ and
   		\begin{itemize}
   			\item $\{av : 1 \leq a \leq 2^{(m-1)/2} +2 \} \subseteq  T_{[6,m,\{2,4,5\}]}$,
   			\item $\{av : 1 \leq a \leq 2^{(m-1)/2} \} \subseteq  T_{[6,m,\{1,2,5\}]}$.
   		\end{itemize}
   		
   	\end{enumerate}		
   \end{lemma}
   \begin{proof}
   	The proof is very similar to that of Lemma \ref{lemma-m121} and omitted here.
   \end{proof}

\vspace{5mm}

\section{Parameters of the codes $\mathcal C_{[6,m,S]}$ and their related codes}\label{sec-duadic}

In this section, we investigate the parameters of the codes $\mathcal C_{[6,m,S]}$ and their related codes.
We begin to consider the case that $m \equiv 1 \pmod 6\geq 7$.
When $S=\{0,4,5\}$, the parameters of  $\C_{[6,m,S]}$ and $\C_{[6,m,\bar S]}$ are treated in the following theorem.

\begin{theorem}\label{thm-m61}
	Let $m \equiv 1 \pmod 6\geq 7$ be an integer. Then $\mathcal C_{[6,m,\{0,4,5\}]}$  and $\mathcal C_{[6,m, \{1,2,3\}]}$ form a pair of odd-like duadic codes with
	parameters $[2^m-1, 2^{m-1}, d]$, where
	\begin{eqnarray*}
		d \geq \left\{
		\begin{array}{ll}
			2^{(m-1)/2}+3 & \mbox{ if } m \equiv 1 \pmod{12}, \\
			2^{(m-1)/2}+1 & \mbox{ if } m \equiv 7 \pmod{12}.
		\end{array}
		\right.
	\end{eqnarray*}
\end{theorem}

\begin{proof}
	It is known that $\C_{[6,m,\{0,4,5\}]}$ and $\C_{[6,m, \{1,2,3\}]}$ form a pair of duadic codes with length $n$ and
	dimension $(n+1)/2$, so they have the same minimum distance $d$.
	
	We only prove the lower bounds on minimum distannce $d$ for $m \equiv 1 \bmod 12 $ as it is similar to prove the  desired conclusion for $m \equiv 7 \bmod 12$. Denote $v=2^{\frac {m-1} 2}-1$. It follows from Lemma \ref{lemma-m121} that $\gcd(v,n) =1$. Let $ \overline{v}$ be the integer satisfying $v \overline{v} \equiv 1 \pmod n$.
	Write $\gamma = \alpha^{\overline{v}}$. It is deduced from Lemma \ref{lemma-m121} that defining set of $\mathcal C_{[6,m,\{0,4,5\}]}$ with respect to $\gamma$ contains the set $\left\{ 1,2,...,2^{(m-1)/2}+2 \right\}$. The lower bound on minimum distance of $\C_{[6,m, \{0,4,5\}]}$ then follows from the BCH bound on the cyclic codes.  This completes the proof.	
\end{proof}

When $S=\{0,4,5\}$, the following theorem provides lower bounds on minimum distances of the dual codes $\C_{[6,m,S]}^\perp$ and $\C_{[6,m,\bar S]}^\perp$.

\begin{theorem}\label{thm-m61-dual}
	Let $m \equiv 1 \pmod 6\geq 7$ be an integer. Then $\C_{[6,m, \{0,4,5\}]}^\perp$ and $\C_{[6,m, \{1,2,3\}]}^\perp$ form a pair of even-like duadic codes with
	parameters $[2^m-1, 2^{m-1}-1, d^\perp]$, where
	\begin{eqnarray*}
		d^\perp \geq \left\{
		\begin{array}{ll}
			2^{(m-1)/2}+4 & \mbox{ if } m \equiv 1 \pmod{12}, \\
			2^{(m-1)/2}+2 & \mbox{ if } m \equiv 7 \pmod{12}.
		\end{array}
		\right.
	\end{eqnarray*}
\end{theorem}

\begin{proof}
	It is easily seen that the defining sets of $\C_{[6,m, \{0,4,5\}]}^\perp$ and $\C_{[6,m, \{1,2,3\}]}^\perp$ with respect to $\alpha$ are
	$\{0\} \cup T_{[6,m,\{0,4,5\}]}$ and $\{0\} \cup T_{[6,m,\{1,2,3\}]}$, respectively. It then follows that  $\C_{[6,m, \{0,4,5\}]}^\perp$ is
	the even-weight subcode of $\C_{[6,m, \{0,4,5\}]}$ and $\C_{[6,m, \{1,2,3\}]}^\perp$ is the even-weight subcode of $\C_{[6,m, \{1,2,3\}]}$.
	The desired conclusion then follows from Theorem \ref{thm-m61}.
\end{proof}

When $S=\{0,4,5\}$, the following theorem presents dimensions and lower bounds on minimum distances of the extended codes $\overline{\C_{[6,m, \{0,4,5\}]}}$ and $\overline{\C_{[6,m,\{1,2,3\}]}}$.

\begin{theorem}\label{thm-m61-extended}
	Let $m \equiv 1 \pmod 6\geq 7$ be an integer. Then the extended codes $\overline{\C_{[6,m, \{0,4,5\}]}}$ and $\overline{\C_{[6,m,\{1,2,3\}]}}$ of
	$\C_{[6,m, \{0,4,5\}]}$ and $\C_{[6,m,\{1,2,3\}]}$ are self-dual and doubly-even,  and they have parameters
	$$[2^m, \ 2^{m-1}, \ \ge 2^{(m-1)/2}+4].$$
\end{theorem}

\begin{proof}
	It is well known that the extended codes of a pair of odd-like binary duadic codes are self-dual if
	the splitting corresponding to  the pair of odd-like binary duadic codes is given by $-1$ \cite[Theorem 6.4.12]{HP03}. As shown earlier,
	$(T_{[6,m,\{0,4,5\}]}, T_{[6,m,\{1,2,3\}]}, -1)$ is a splitting of $\Bbb Z_{2^m-1}$. Consequently, the extended codes
	$\overline{\C_{[6,m,\{0,4,5)]}}$ and $\overline{\C_{[6,m,\{1,2,3)]}}$  are self-dual.
	It then follows from \cite[Theorem 6.5.1]{HP03} that the Hamming weight of each codeword
	in $\overline{\C_{[6,m,\{0,4,5\}]}}$ and $\overline{\C_{[6,m,\{1,2,3\}]}}$ is divisible by $4$. The remaining conclusions follow from Theorem \ref{thm-m61}.
\end{proof}

When $S=\{0,3,5\}$ or $\{0, 2,3\}$, we have the following theorem on parameters of the codes $\C_{[6,m,S]}$ and $\C_{[6,m,\bar S]}$ and their related codes.

\begin{theorem}\label{thm-m61-2}	
	Let $m \equiv 1 \pmod 6\geq 7$ be an integer and suppose that $S=\{0,3,5\}$ or $\{0, 2,3\}$.
\begin{enumerate}
  \item The odd-like duadic codes $\C_{[6,m,S]}$ and $\C_{[6,m,\bar S]}$ have
	parameters $$[2^m-1, \ 2^{m-1}, \ \ge 2^{(m-1)/2}+3].$$
  \item The dual codes $\C_{[6,m,S]}^\perp$ and $\C_{[6,m,\bar S]}^\perp$ form a pair of even-like duadic codes with
	parameters $$[2^m-1, \ 2^{m-1}-1, \ \ge 2^{(m-1)/2}+4].$$
  \item The extended codes $\overline{\C_{[6,m,S]}}$ and $\overline{\C_{[6,m,\bar S]}}$ are self-dual and doubly-even, and they have parameters $$[2^m, \ 2^{m-1}, \ \ge 2^{(m-1)/2}+4].$$
\end{enumerate}
\end{theorem}

\begin{proof}
	It is very similar to those of Theorems \ref{thm-m61}, \ref{thm-m61-dual}, \ref{thm-m61-extended} and omitted here.
\end{proof}

When $m \equiv 3 \pmod 6$, we have the following theorem on parameters of  the codes $\C_{[6,m,S]}$ and $\C_{[6,m,\bar S]}$ and their dual and extended codes.
It can be similarly proved and we omit the details here.

\begin{theorem}\label{thm-m63-dual}
	Let $m \equiv 3 \pmod 6$ be an integer.

When $S=\{0,1,5\}$, we have the following.
\begin{enumerate}
  \item  The odd-like duadic codes $\mathcal C_{[6,m,S]}$  and $\mathcal C_{[6,m, \bar S]}$ have
	parameters $[2^m-1, 2^{m-1}, d]$, where
	\begin{eqnarray*}
		d \geq \left\{
		\begin{array}{ll}
			2^{(m-1)/2}+1 & \mbox{ if } m \equiv 3 \pmod{12}, \\
			2^{(m-1)/2}+3 & \mbox{ if } m \equiv 9 \pmod{12}.
		\end{array}
		\right.
	\end{eqnarray*}
  \item The codes $\C_{[6,m,S]}^\perp$ and $\C_{[6,m, \bar S]}^\perp$ form a pair of even-like duadic codes with
	parameters $[2^m-1, 2^{m-1}-1, d^\perp]$, where
	\begin{eqnarray*}
		d^\perp \geq \left\{
		\begin{array}{ll}
			2^{(m-1)/2}+2 & \mbox{ if } m \equiv 3 \pmod{12}, \\
			2^{(m-1)/2}+4 & \mbox{ if } m \equiv 9 \pmod{12}.
		\end{array}
		\right.
	\end{eqnarray*}
  \item The extended codes $\overline{\C_{[6,m,S]}}$ and $\overline{\C_{[6,m, \bar S]}}$ are self-dual and doubly-even, and they have parameters
	$$[2^m, \ 2^{m-1}, \ \ge 2^{(m-1)/2}+4].$$
\end{enumerate}

When $S=\{0,1,4\}$ or $\{0,2,5\}$, we have the following.

\begin{enumerate}
  \item The odd-like duadic codes $\C_{[6,m,S]}$ and $\C_{[6,m, \bar S]}$ have
	parameters $$[2^m-1, \ 2^{m-1}, \ \ge 2^{(m-1)/2}+1 ].$$
  \item The codes $\C_{[6,m,S]}^\perp$ and $\C_{[6,m, \bar S]}^\perp$ form a pair of even-like duadic codes with
	parameters $$[2^m-1, \ 2^{m-1}-1, \ \ge 2^{(m-1)/2}+2].$$
  \item The extended codes $\overline{\C_{[6,m,S]}}$ and $\overline{\C_{[6,m, \bar S]}}$ are self-dual and doubly-even, and they have parameters $$[2^m, \ 2^{m-1}, \ \ge 2^{(m-1)/2}+4].$$
\end{enumerate}
\end{theorem}

When $m \equiv 5 \pmod 6$, we have the following theorem on parameters of  the codes $\C_{[6,m,S]}$ and $\C_{[6,m,\bar S]}$ and their dual and extended codes.
It can be similarly proved and we omit the details here.

\begin{theorem}\label{thm-m65}
	Let $m \equiv 5 \pmod 6$ be an integer.

When $S=\{0, 1, 2\}$, we have the following.

\begin{enumerate}
  \item The odd-like duadic codes $\C_{[6,m,S]}$ and $\C_{[6,m, \bar S]}$ have
	parameters $[2^m-1, 2^{m-1}, d]$, where
	\begin{eqnarray*}
		d \geq \left\{
		\begin{array}{ll}
			2^{(m-1)/2}+1 & \mbox{ if } m \equiv 5 \pmod{12}, \\
			2^{(m-1)/2}+3 & \mbox{ if } m \equiv 11 \pmod{12}.
		\end{array}
		\right.
	\end{eqnarray*}
  \item The codes $\C_{[6,m,S]}^\perp$ and $\C_{[6,m, \bar S]}^\perp$ form a pair of even-like duadic codes with
	parameters $[2^m-1, 2^{m-1}-1, d^\perp]$, where
	\begin{eqnarray*}
		d^\perp \geq \left\{
		\begin{array}{ll}
			2^{(m-1)/2}+2 & \mbox{ if } m \equiv 5 \pmod{12}, \\
			2^{(m-1)/2}+4 & \mbox{ if } m \equiv 11 \pmod{12}.
		\end{array}
		\right.
	\end{eqnarray*}
  \item The extended codes $\overline{\C_{[6,m,S]}}$ and $\overline{\C_{[6,m, \bar S]}}$ are self-dual and doubly-even, and they have parameters
	$$[2^m, \ 2^{m-1}, \ \ge 2^{(m-1)/2}+4].$$
\end{enumerate}

When $S=\{0, 1, 3\}$, we have the following.
\begin{enumerate}
  \item The odd-like duadic codes $\C_{[6,m, S]}$ and $\C_{[6,m,\bar S]}$ have parameters $$[2^m-1,  \ 2^{m-1},  \ \ge 2^{(m-1)/2}+3 ].$$
  \item The codes $\C_{[6,m,S]}^\perp$ and $\C_{[6,m,\bar S]}^\perp$ form a pair of even-like duadic codes with parameters $$[2^m-1, \  2^{m-1}-1, \  \ge 2^{(m-1)/2}+4].$$
  \item The extended codes $\overline{\C_{[6,m,S]}}$ and $\overline{\C_{[6,m,\bar S]}}$ are self-dual and doubly-even, and they have parameters $$[2^m, \  2^{m-1},  \ \ge 2^{(m-1)/2}+4].$$
\end{enumerate}

When $S=\{0, 3, 4\}$, we have the following.

\begin{enumerate}
  \item The odd-like duadic codes $\C_{[6,m,S]}$ and $\C_{[6,m,\bar S]}$ have parameters $$[2^m-1, \ 2^{m-1}, \ \ge 2^{(m-1)/2}+1].$$
  \item The codes $\C_{[6,m,S]}^\perp$ and $\C_{[6,m,\bar S]}^\perp$ form a pair of even-like duadic codes with parameters $$[2^m-1, \ 2^{m-1}-1, \ \ge 2^{(m-1)/2}+2].$$
  \item The extended codes $\overline{\C_{[6,m,S]}}$ and $\overline{\C_{[6,m,\bar S]}}$ are self-dual and doubly-even, and they have parameters $$[2^m, \ 2^{m-1}, \ \ge 2^{(m-1)/2}+4].$$
\end{enumerate}
\end{theorem}

It can be seen that the lower bounds on  minimum distances of the odd-like and even-like duadic codes
 developed in this paper are very close to the square root bound.

\begin{example}\label{exam-J102}
	Let $m=7$ and  let $\alpha$ be a generator of  $\Bbb F_{2^7}^*$ with $\alpha^{7}+\alpha+1 = 0$.
The parameters of $\C_{[6,7,S]}$ and $\C_{[6,7,\bar S]}$ and their dual codes are documented in Table \ref{Tab}.
\begin{table}
     \centering
     \caption{Parameters of $\C_{[6,7,S]}$ and $\C_{[6,7,\bar S]}$ and their dual codes} \label{Tab}
     \begin{tabular}{|c|c|c|c|c|c|c|c|c|c|c|}
      \hline
      $S$  & $\C_{[6,7,S]}$  &  $\C_{[6,7, S]}^\perp$ &&  $\bar S$  &  $\C_{[6,7,\bar S]}$  &  $\C_{[6,7, \bar S]}^\perp$\\ \hline
      $\{0, 2,3\}$ & $[127, 64, 15]$ &     $[127, 63, 20]$ &  & $\{1,4,5\}$ & $[127, 64, 15]$ &  $[127, 63, 20]$   \\ \hline
      $\{0,3,5\}$ & $[127, 64, 19]$ &    $[127, 63, 20]$ &  & $\{1,2,4\}$ & $[127, 64, 19]$ &    $[127, 63, 20]$  \\ \hline
      $\{0,4,5\}$ & $[127, 64, 15]$ &    $[127, 63, 16]$ &  & $\{1,2,3\}$ & $[127, 64, 15]$ &    $[127, 63, 16]$   \\ \hline
     \end{tabular}
    \end{table}
\end{example}

\section{Summary and concluding remarks}\label{sec-Con}

In this paper, we investigated the parameters of several families of binary duadic codes $\mathcal \C_{[6,m,S]} $ and their dual and extended codes,
 which partially solves one case (i.e. $r=6$) of the open problem proposed in \cite{LLD}.
The dimensions of these codes were determined explicitly and lower bounds on their minimum distances were presented.

The only known binary duadic codes with a good minimum distance or a good lower bound on their minimum distances
are the following:
\begin{itemize}
\item Binary quadratic residue codes with parameters $[n, (n+1)/2, d]$ and their even-weight
          subcodes, where $d^2 \geq n$ and $n \equiv \pm 1 \pmod{8}$ is a prime \cite[Section 6.6]{HP03}.
\item The punctured binary Reed-Muller codes of order $(m-1)/2$ which has parameters $[2^m-1, 2^{m-1}, 2^{(m+1)/2}-1]$, where $m$
          is odd \cite{AK}.
\item Binary duadic codes $\C_{[2,7,\{0\}]}$ and $\C_{[2,7,\{1\}]}$ constructed in \cite{TD}, and binary duadic codes $\C_{[4,7,\{0,1\}]}$ and $\C_{[4,7,\{2,3\}]}$ constructed in \cite{LLD}
(Lower bounds on their minimum distances are close to the square root bound).
\end{itemize}

 The binary duadic codes constructed in this paper are not identical with
the family of binary quadratic residue codes, as $2^m-1$ is composite for many odd $m$.
Some known binary duadic codes of length $127$ are listed in Table \ref{Tab-known}, where $\text{PRM}_2(3, 7)$ means the punctured binary
Reed-Muller code of length $127$ and order $3$ \cite{AK}. From Tables \ref{Tab} and \ref{Tab-known}, one can always find at least one subset $S$ such that
the parameters of $\C_{[6,7,S]}$ are different from those of punctured binary Reed-Muller codes and the binary duadic codes presented in \cite{TD} and \cite{LLD}.
\begin{table}
     \centering
     \caption{Parameters of some known duadic codes } \label{Tab-known}
     \begin{tabular}{|c|c|c|c|c|}
      \hline
      Code & Parameters & Dual code & Parameters & Reference \\ \hline
      $\text{PRM}_2(3, 7)$ & $[127, 64, 15]$ & $\text{PRM}_2(3, 7)^\perp$ & $[127, 63, 16]$  & \cite{AK}\\ \hline
      $\C_{[2,7,\{0\}]}$  &  $[127, 64, 19]$ & $\C_{[2,7,\{0\}]}^\perp$  & $[127, 63, 20]$  & \cite{TD} \\ \hline
       $\C_{[2,7,\{1\}]}$  & $[127, 64, 19]$ & $\C_{[2,7,\{1\}]}^\perp$  & $[127, 63, 20]$  & \cite{TD} \\ \hline
       $\C_{[4,7,\{0,1\}]}$  & $[127, 64, 15]$ & $\C_{[4,7,\{0,1\}]}^\perp$  & $[127, 63, 20]$  & \cite{LLD} \\ \hline
       $\C_{[4,7,\{2,3\}]}$  & $[127, 64, 15]$ & $\C_{[4,7,\{2,3\}]}^\perp$  & $[127, 63, 20]$  & \cite{LLD} \\ \hline
     \end{tabular}
    \end{table}

There are only a few families of binary duadic
codes whose minimum distances have a lower bound close
to the square-root bound.
It was shown that the lower bounds on  minimum distances of the binary duadic codes constructed in this paper are close to the square root bound.
Cyclic codes have many important applications in data storage and communication systems.
However, it is not easy to construct cyclic codes with a large
dimension and good minimum distance when the length $n$
is the product of some small primes \cite{XiongZhang}.
It would be interesting to construct more  families of binary duadic
codes with good minimum distances.

\end{document}